\documentclass{LMCS}

\usepackage{amssymb,enumerate,hyperref}
\usepackage{latexsym}

\newtheorem{definition}[thm]{Definition}

\newcommand{\tuple}[1]{\left(#1\right)}

\def\sm{\setminus}

\newcommand{\ub}{\mathit{ub}}

\newcommand{\ra}{\rightarrow}
\newcommand{\Ps}[1]{\mathbb{P}( #1 )}
\newcommand{\Char}{\mathit{Char}}
\newcommand{\gt}{\mbox{game-type}}
\newcommand{\code}{\mbox{Code}}
\newcommand{\trun}{\mbox{trun}}
\newcommand{\gcode}{\mbox{gcode}}

\newcommand{\nat}{\mathbb{N}}

\newcommand{\om}{\omega}
\def\vp{\varphi}
\def\mG{\mathcal{G}}
\def\rar{\rightarrow}

\def\s{\subseteq}

\newcommand{\MLO}{\mathit{MLO}}
\def\mM{\mathcal{M}}
\def\mI{\mathcal{I}}
\def\mN{\mathcal{N}}
\def\MTh{\mathit{MTh}}
\def\HF{\mathit{type}}
\def\Th{\mathit{Th}}

\def\Form{\mathfrak{Form}}
\def\fS{\mathfrak{S}}
\def\fH{\mathfrak{H}}
\def\nek{\ldots}
\newcommand{\qd}[1]{\mathop{\rm qd}(#1)}

\def\la{\lambda}

\def\win{\mathit{Win}}

\def\doi{5 (2:5) 2009}
\lmcsheading%
{\doi}
{1--19}
{}
{}
{Okt.~27, 2007}
{Apr.~27, 2009}
{}   

\begin{document}

\title{The Church Problem  for Countable Ordinals}

\author{ Alexander Rabinovich}
\address{The Blavatnik School  of Computer Science, Tel Aviv University, Israel 69978}
\email{rabinoa@post.tau.ac.il}

\keywords{Church's problem, decidability, determinacy, composition method}
\subjclass{F.4.1}

\begin{abstract}  A fundamental  theorem  of B\"{u}chi and
Landweber shows that
 the Church synthesis problem  is computable.
  B\"{u}chi and
Landweber reduced the Church Problem to problems about
$\omega$-games  and
 used the determinacy of such
games as one of the main tools  to show its computability.
We consider a natural generalization of the Church problem to
countable ordinals and investigate  games of arbitrary countable
length. We prove that determinacy and decidability parts of the
B\"{u}chi and Landweber theorem hold for all countable ordinals
and that its full extension holds for all ordinals $<\om^\om$.
\end{abstract}

 \maketitle

\section{Introduction}
Two fundamental results of classical automata theory are
decidability of the monadic second-order logic of order (MLO) over
 $\omega=(\nat,<)$ and computability of the Church synthesis problem.
These results have provided the underlying mathematical framework
for the development of formalisms for the description of
interactive systems and their desired properties,  the algorithmic
verification  and the automatic synthesis of correct
implementations from logical specifications, and  advanced
algorithmic techniques that are now embodied in industrial  tools
for verification and validation.

In order to prove  decidability of  the monadic theory of
$\omega$, B\"{u}chi introduced finite automata over
$\omega$-words. He provided a computable reduction from formulas
to finite automata.

B\"{u}chi generalized the concept of an automaton to allow
automata to ``work" on words of any countable length (ordinal) and
used this to show that the MLO-theory of any countable ordinal is
decidable (see \cite{buchi:siefkes}).

What is known as the ``Church synthesis problem'' was first posed
by A. Church in \cite{church} for the case of $(\om,<)$. The
Church  problem is much more complex than the decidability problem
for $MLO$.  Church uses the language of automata theory. It was
McNaughton (see \cite{mcnaughton}) who first observed that the
Church problem  can be equivalently phrased in game-theoretic
language.

 Let $\alpha>0$ be an ordinal and let  $\vp(X_1,X_2)$ be a formula, where $X_1$ and $X_2$
 are set (monadic predicate) variables.
The \emph{McNaughton game} $\mG_\vp^\alpha$ is defined as follows.
\begin{enumerate}[(1)]
\item The game is played by two players, called Player I and Player II.
\item A \emph{play} of the game has $\alpha$ rounds.
\item At round $\beta<\alpha$:
first, Player I  chooses $\pi_{X_1}(\beta)\in \{0,1\}$; then,
Player II  chooses $\pi_{X_2}(\beta)\in  \{0,1\}$.
\item By the end of the play $\pi_{X_1},\pi_{X_2}:\alpha\rar \{0,1\}$
have been constructed. Set $$P_\pi:= \pi_{X_1}^{-1}(1),{} Q_\pi:=
\pi_{X_2}^{-1}(1).$$
\item Then, Player I  wins the play if $(\alpha,<)\models \vp(P_\pi,Q_\pi)$;
otherwise, Player II  wins the play.
\end{enumerate}
What we want to know is: Does either one of the players have a
\emph{winning strategy} in $\mG_\vp^\alpha$? If so, which one?
That is, can Player I  choose his moves so that, whatever way
Player II responds we have $ \vp(P_\pi,Q_\pi)$? Or can Player II
respond to Player I's moves in a way that ensures the opposite?

Since at round $\beta<\alpha$, Player I has access only to $Q_\pi
\cap [0,\beta)$ and Player II  has access only to $P_\pi\cap
[0,\beta]$, it seems that the following formalizes  well the
notion of a strategy in this game:
\begin{definition}[Causal operator] Let $\alpha$ be an ordinal, $f:\Ps{\alpha}\rar \Ps{\alpha} $ maps the subsets of $\alpha$ into the subsets of $\alpha$.
We call $f$ \emph{causal} (resp. \emph{strongly causal}) iff for
all $P,P'\s \alpha$ and $\beta<\alpha$, if
\[P\cap [0,\beta] = P'\cap [0,\beta]\quad\hbox{(resp. $P\cap [0,\beta) = P'\cap
[0,\beta)$)} \quad\hbox{implies}\quad
f(P)\cap [0,\beta] = f(P')\cap [0,\beta]\ .
\]
That is, if $P$ and $P'$ agree \emph{up to and including} (resp.
\emph{up to}) $\beta$, then so do $f(P)$ and $f(P')$.
\end{definition}
So a winning strategy for Player I is a strongly causal $f:\Ps{\alpha}\rar
\Ps{\alpha}$ such that for every $P\s \alpha$, $(\alpha,<)\models
 \vp(f(P),P)$;
 a winning strategy for Player II is a causal $f:\Ps{\alpha}\rar
\Ps{\alpha}$ such that for every $P\s \alpha$, $(\alpha,<)\models
\neg \vp(P,f(P))$.

It is clear that if Player I  has a winning strategy in
$\mG_\vp^{\alpha}$, then $\alpha\models \forall X_2\exists
X_1\vp$. It is also  easy to see that  $\alpha\models \forall
X_2\exists X_1\vp$ does not imply that Player I has a winning
strategy.

This leads to
\begin{definition}[Game version of the Church  problem]\label{dfn:Church synthesis problem}
Let $\alpha$ be an ordinal. Given a $\MLO$-formula $\vp(X_1,X_2)$, decide
whether Player I  has a winning strategy in $\mG_\vp^\alpha$.
\end{definition}
From now on, we will use ``formula" for ``$\MLO$-formula" unless stated otherwise.

To simplify notations,  games and the Church  problem were
previously defined for formulas with two free variables $X_1$ and
$X_2$. It is easy to generalize all definitions and results to
formulas $\psi(X_1,\dots , X_m, Y_1, \dots Y_n)$ with many
variables. In this generalization at round  $\beta$, Player I
chooses values for $X_1(\beta), \dots , X_m(\beta)$, then  Player
II  replies by choosing  values for  $Y_1(\beta), \dots ,
Y_n(\beta)$.
 Note
that, strictly speaking, the input to the Church problem  is not
only a formula, but a formula plus a partition of its
free variables to Player I's variables and  Player II's variables.

In \cite{BL69}, B\"{u}chi and Landweber prove the computability of
the Church problem in $\om=(\nat,<)$. Even more importantly, they show
that in the case of $\om$ we {can}  restrict ourselves to
 \emph{definable strategies}, that is  causal (or strongly causal)
operators computable by finite state automata or, equivalently,
definable in $\MLO$. (Recall that  $F:\Ps{\nat}\rar \Ps{\nat} $ is definable in $\MLO$ if there is  an $\MLO$  formula $\psi(X,Y)$ such
that  for every $P,Q\subseteq \nat$ we
have $P=F(Q)$ iff  $\om\models \psi(Q,P)$.)
\begin{thm}[B\"{u}chi-Landweber, 1969]\label{th:bl-game} Let $\vp(\bar{X},\bar{Y})$ be a formula,
where $\bar{X}$ and $\bar{Y}$ are  disjoint lists of variables.
Then:
\begin{enumerate}[\hbox to6 pt{\hfill}]
\item\noindent{\hskip-0 pt\bf Determinacy:}\ One of the players has a 
  winning strategy in the game $\mG_\vp^\om$.\medskip

\item\noindent{\hskip-0 pt\bf Decidability:}\ It is decidable
  \emph{which} of the players has a winning strategy.\medskip

\item\noindent{\hskip-0 pt\bf Definable strategy:}\ The player who has 
  a winning strategy, also has a \emph{definable} winning strategy.\medskip

\item\noindent{\hskip-0 pt\bf Synthesis algorithm:}\ We can compute a
  formula $\psi(\bar{X},\bar{Y})$ that defines (in $\om$) a winning 
  strategy for the winning player in $\mG_\vp^\om$.
%\footnote{ If $\psi$ is to define a strategy for
%Player I, then we should take $\bar{Y}$ as the domain variables,
%so to speak.}
\end{enumerate}
\end{thm}
\noindent After stating their main theorem,
 B\"{u}chi and Landweber write:
\begin{quote}
 ``We hope to present elsewhere a corresponding extension of [our main theorem] from $\om$
to any countable ordinal.''
\end{quote}
 However,  despite the fundamental role of the Church problem, no such extension is even mentioned in
 a later book by
 B\"{u}chi and Siefkes
 \cite{buchi:siefkes}, which summarizes the theory of finite automata over words of countable ordinal
 length.

In \cite{RS06}, we provided  a counter-example to a full extension
of the B\"{u}chi-Landweber   theorem to $\alpha\geq \om^\om$.
 Let $\vp (Y)$, $\psi({Y})$ be formulas and $\mM$  be a structures.
We say that $\psi$ \emph{selects} (or, is a \emph{selector} for)
$\vp$ \emph{in} $\mM$ iff:
\begin{enumerate}[(1)]
\item $\mM \models \exists^{\le 1}{Y}\psi({Y})$ (i.e., there is at
most one $Y$ that satisfies $\psi$),
\item $\mM \models \forall {Y}(\psi({Y})\rar \vp({Y}))$, and
\item $\mM \models \exists {Y}\vp({Y}) \rar \exists {Y}\psi({Y})$.
\end{enumerate}
In \cite{RS06}, we proved that for every ordinal $\alpha\geq
\om^\om$ there is an $\MLO$  formula $\psi_\alpha(Y)$ such that
$\alpha\models \exists Y\psi_\alpha$; however,  there is no $\MLO$
formula that selects $\psi_\alpha$ in $\alpha$.

Now,  consider McNaughton game of length $\alpha$  for
$\vp(X_1,X_2)$ defined as $\psi_\alpha(X_1)\wedge X_2=X_2$
 (so,  we ignore $X_2$ and this can be considered as a one-player game).
 Player I  cannot have a {definable} strategy in
$\mG_\vp^\alpha$. Indeed, if $\chi$ could define such a strategy,
then $\exists X_2(X_2=\emptyset \wedge \chi)$  would select
$\psi_{\alpha}(X_1)$ over $\alpha$. On the other hand, Player I
does win this game: she simply plays a \emph{fixed} $X_1$ that
satisfies $\psi_\alpha(X_1)$ over $\alpha$, ignoring Player II
moves. Hence, the definability and the synthesis  parts of the
B\"{u}chi-Landweber theorem fail for every $\alpha\geq \om^\om$.
 Shomrat  \cite{Sho07}  proved
\begin{thm}\label{th:intr-shomrat} The B\"{u}chi-Landweber theorem holds for an ordinal $\alpha$ if and only if $\alpha<\omega^\omega$.
\end{thm}
The proof in \cite{Sho07} is very long and requires the
development of an extensive game-theoretic apparatus. Moreover,
the technique of \cite{Sho07}  cannot be extended to games of
length $\geq \om^\om$. In this article, we provide a simple proof
of  Theorem \ref{th:intr-shomrat}.

Our main results show that the determinacy and the decidability
parts of the B\"{u}chi-Landweber Theorem hold for every countable
ordinal.
\begin{thm}[Main] Let $\alpha$ be a countable ordinal and $\vp(\bar{X},\bar{Y})$ be a
formula.Then:
\begin{enumerate}[\hbox to6 pt{\hfill}]
\item\noindent{\hskip-0 pt\bf Determinacy:}\ One of the players has a winning strategy in the game $\mG_\vp^\alpha$.\medskip

\item\noindent{\hskip-0 pt\bf Decidability:}\ It is decidable \emph{which} of the players has a winning strategy in $\mG_\vp^\alpha$.
\end{enumerate}
\end{thm}
\noindent Our proof uses both game theoretical techniques and the
``composition method" developed by Feferman-Vaught, Shelah and
others (see, e.g. \cite{Sh75}).

The article is organized as follows. The next section recalls
standard definitions about monadic logic of order, summarizes
elements of the composition method and reviews  known facts about
the monadic theory of countable ordinals. In Section
\ref{sec:games}, we introduce game-types, define games on game
types and show  that these game are  a special case of McNaughton
games.  Section \ref{sect:comp-gt}  contains  the main technical
lemmas of the paper and shows that the role of the game types is
similar to the role of monadic types in the composition method. In
Section \ref{sect:sho}, we prove that the B\"{u}chi-Landweber
theorem holds in its entirety for all ordinals
$\alpha<\omega^\omega$. In Sections \ref{sect:deter} and
\ref{sect:dec}, it is proved that the determinacy and decidability
parts of the B\"{u}chi-Landweber theorem hold for all countable
ordinals.
Section \ref{sect:mlo-char} provides an  MLO-characterisation of the winner
and  shows that for every
formula $\varphi(X,Y)$  there exists
an  $\MLO$ sentence $\psi$   such that  for every   countable
ordinal  $\alpha$:  Player I wins $\mG_\vp^\alpha$ iff $
\alpha\models \psi$.
 Section \ref{sect:synth} addresses the problem whether
the  winner has a definable winning strategy. We were unable to
show that this problem is decidable for ordinals $\geq \om^\om$;
however, we reduced  the decidability of this problem for
countable ordinals $\geq \om^\om$ to the case of games of length
 $\om^\om$. Section  \ref{sect:concl} contains a conclusion and states
some open problems.

\section{Preliminaries and Background}
\subsection{Notations and terminology}
We use $n,k,l,m,p,q$ for natural numbers and
$\alpha,\beta,\gamma,\delta$ for ordinals.  We use $\nat$ for the
set of natural numbers and  $\om$ for the first infinite ordinal.
We write $\alpha+\beta$, $\alpha\beta$, $\alpha^\beta$ for the
sum, multiplication and exponentiation, respectively, of ordinals
$\alpha$ and $\beta$. We use the expressions ``\emph{chain}'' and
``\emph{linear order}'' interchangeably.

We use  $\Ps{A}$ for the set of subsets of $A$.

\subsection{The Monadic Logic of Order ($\MLO$)}
\subsubsection{Syntax}
 The syntax of the monadic second-order logic of order - $\MLO$ has
in its vocabulary {\it individual\/} (first order) variables
$t_1,t_2\nek$, monadic  {\it second-order\/} variables
$X_1,X_2\nek$  and one binary relation $<$ (the order).

 { Atomic formulas} are
of the form $X(t)$ and $t_1<t_2$. { Well formed formulas} of the
monadic logic $\MLO$ are obtained from atomic formulas using
Boolean connectives $\neg, \vee, \wedge,\to$ and the first-order
quantifiers $\exists t$ and $\forall t$, and the second-order
quantifiers $\exists X$  and $\forall X$. The quantifier depth of
a formula $\varphi$ is denoted by $\qd{\varphi}$.

We use upper case letters $X$, $Y$, $Z$,...  to denote
second-order variables; with an overline, $\bar{X}$, $\bar{Y}$,
etc., to denote finite tuples of variables.
\subsubsection{Semantics} A \emph{structure} is a tuple $\mM:=(A,<^\mM, \bar{P}^\mM)$
where: $A$ is a non-empty set, $<^\mM$ is a binary relation on
$A$, and $\bar{P}^\mM:= \tuple{ P_1^\mM,\ldots,P_{l}^\mM}$ is a
\emph{finite} tuple of sub\emph{sets} of $A$.

If $\bar{P}^\mM$ is a tuple of $l$ sets, we call $\mM$ an
$l$-\emph{structure}. If $<^\mM$ linearly orders $A$, we call $\mM$
an $l$-\emph{chain}.

Suppose $\mM$ is an $l$-structure and $\vp$ a formula with
free-variables among $X_1,\ldots,X_{l}$. We define the relation
$\mM \models \vp$ (read: $\mM$ \emph{satisfies} $\vp$) as usual,
 understanding that the second-order quantifiers range over
\emph{subsets} of $A$.

Let $\mM$ be an $l$-structure. The \emph{monadic theory} of $\mM$,
$\MTh(\mM)$, is the set of all formulas with free-variables among
$X_1,\ldots,X_{l}$ satisfied by $\mM$.

From now on, we omit the superscript  in `$<^\mM$' and
`$\bar{P}^\mM$'. We often write $(A,<)\models \vp(\bar{P})$
meaning $(A,<,\bar{P})\models \vp$.
\begin{definition}[$\MLO$ Definable Function]Let  $\mM:=(A,<)$ be a chain.  %We say that a
A function $F:\Ps{A}^n\rightarrow \Ps{A}^m$ is $\MLO$-\emph{definable }in $\mM$ if there is an $\MLO$ formula $\vp(X_1,
\dots, X_n,Y_1,\dots, Y_m)$ such that
 $\mM\models\vp(P_1,\dots,
P_n,Q_1,\dots Q_m)$ iff $\tuple{Q_1,\dots ,Q_m}=F(P_1,\dots ,P_n)$.
\end{definition}
\subsection{Elements of the composition method}\label{subsection:composition method}
Our proofs make use of the technique known as the composition method
developed by Feferman-Vaught and   Shelah \cite{FV59,Sh75}. To fix
notations and to aid the reader unfamiliar with this technique, we
briefly review the definitions and results that we require. A more
detailed presentation can be found in \cite{thomas:ef+composition}
or \cite{gurevich:monadic second-order theories}.
\subsubsection{Hintikka formulas and $n$-types}
Let $n,l\in \nat$. We denote by $\Form^n_{l}$ the set of formulas
with free variables among $X_1,\ldots,X_{l}$ and of quantifier depth
$\le n$.
\begin{definition} Let $n,l\in \nat$ and let $\mM, \mN$ be  $l$-structures.
The $n$-\emph{theory} of $\mM$ is
$$\Th^{n}(\mM):= \{\vp\in \Form^n_{l}\mid \mM\models \vp\}.$$
If $\Th^{n}(\mM)=\Th^{n}(\mN)$, we say that $\mM$ and $\mN$ are
$n$-\emph{equivalent} and write $\mM\equiv^{n}\mN$.
\end{definition}
Clearly, $\equiv^{n}$ is an equivalence relation. For any $n\in \nat$
and $l>0$, the set $\Form^n_{l}$ is infinite. However, it contains
only finitely many semantically distinct formulas. So, there are
finitely many $\equiv^n$-equivalence classes of $l$-structures. In
fact, we can compute characteristic sentences for the
$\equiv^n$-classes:
\begin{lem}[Hintikka Lemma]\label{lemma:Hintikka formulas}
For $n,l\in \nat$, we can compute a \emph{finite}  set $\Char^n_l\s
\Form^n_{l}$ such that:
\begin{enumerate}[$\bullet$]
\item For every $\equiv^n$-equivalence class $A$ there is a \emph{unique}
$\tau\in \Char^n_l$ such that for every $l$-structure $\mM$:
$\mM\in A$ iff $\mM\models \tau$.

\item
 Every $\MLO$ formula  $\varphi(X_1,\dots X_l)$ with $\qd{\varphi}
\leq k$ is equivalent to a (finite) disjunction of characteristic
formulas from $Char^k_l$. Moreover, there is an algorithm which
for every formula $\varphi(X_1,\dots X_l)$ computes a  finite set
$G\subseteq \Char^{\qd{\varphi}}_l$ of characteristic formulas,
such that $\varphi$ is equivalent to the disjunction of all the
formulas in $G$.
\end{enumerate}
Any member of $ Char^k_l$ we call a $(k,l)$-\emph{Hintikka formula}
or $(k,l)$-characteristic formula. We  use  $\tau$, $\tau_i$,
$\tau^j$ to range over the characteristic formulas and $G, G_i, G'$
to  range over sets of  characteristic formulas.
\end{lem}
\begin{definition}[$n$-Type] For $n,l\in \nat$ and an $l$-structure  $\mM$, we denote by
$\HF_n(\mM)$ the unique member of $\Char^n_l$ satisfied by $\mM$
and call it the $n$-\emph{type} of $\mM$.
\end{definition}
Thus, $\HF_n(\mM)$ determines $\Th^n(\mM)$ and, indeed,
$\Th^n(\mM)$ is computable from $\HF_n(\mM)$.
\subsubsection{The ordered sum of chains and of $n$-types}
\begin{definition}\hfill
\begin{enumerate}[(1)]
\item Let $l\in \nat$, $\mI:=(I,<^\mI)$ a chain and $\fS:=\tuple{
\mM_\alpha \mid \alpha\in I}$ a sequence of $l$-chains. Write
$\mM_\alpha:= (A_\alpha,<^\alpha,{P_1}^\alpha, \dots,{P_l}^\alpha
)$ and assume  that $A_\alpha\cap A_\beta=\emptyset$ whenever $\alpha\ne
\beta$ are in $I$. The \emph{ordered sum} of $\fS$ is the
$l$-chain
$$\sum_\mI\fS:= (\bigcup_{\alpha\in I}A_\alpha,<^{\mI,\fS},{\bigcup}_{\alpha\in I}{P_1}^\alpha, \dots , \bigcup_{\alpha\in I}{P_l}^\alpha),$$
where:
\[\hbox{if $\alpha,\beta\in I$, $a\in A_\alpha$, $b\in A_\beta$,
then $b<^{\mI,\fS} a$ iff $\beta<^\mI\alpha$ or $\beta=\alpha$ and
$b<^\alpha a$.}
\]
If the domains of the $\mM_\alpha$'s are not disjoint, replace
them with isomorphic $l$-chains that have disjoint domains, and
proceed as before.

\item If for all $\alpha\in I$, $\mM_\alpha$ is isomorphic to $\mM$ for some
\emph{fixed} $\mM$, we denote $\sum_\mI\fS$ by $\mM\times \mI$.

\item If $\mI=(\{0,1\},<)$ and $\fS=\tuple{\mM_0,\mM_1}$, we
denote $\sum_\mI\fS$ by $\mM_0 + \mM_1$.
\end{enumerate}
\end{definition}

The next proposition says that taking ordered sums preserves
$\equiv_n$-equivalence.
\begin{prop}\label{prop:ordered sum preserves equivalence} Let $n,l\in \nat$. Assume:
\begin{enumerate}[\em(1)]
\item $( I,<^\mI)$ is a linear order,
\item $\tuple{ \mM_\alpha^0 \mid \alpha\in I}$ and
$\tuple{ \mM_\alpha^1 \mid \alpha\in I}$ are sequences of
$l$-chains, and
\item for every $\alpha\in I$, $\mM_\alpha^0 \equiv^{n} \mM_\alpha^1$.
\end{enumerate}
Then,  $\sum_{\alpha\in I}\mM_\alpha^0 \equiv^{n} \sum_{\alpha\in
I}\mM_\alpha^1$.
\end{prop}
This allows us to define the sum of formulas in $\Char^{n}_l$ with
respect to any linear order.
\begin{definition}\label{dfn:sum of n types} \hfill
\begin{enumerate}[(1)]
\item  Let $n,l\in \nat$, $\mI:=(I,<^\mI)$ a chain,
$\fH:=\tuple{ \tau_\alpha\mid \alpha\in I}$ a sequence of
$(n,l)$-Hintikka formulas. The \emph{ordered sum} of $\fH$,
(notations  $\sum_\mI\fH$ or $ \sum_{\alpha\in
\mI}\tau_{\alpha}$), is \emph{an}  element $\tau$ of $\Char^n_l$
such that:\hfill\break
if $\fS:=\tuple{ \mM_\alpha \mid \alpha \in I}$ is a sequence
of $l$-chains and $\HF_{n}(\mM_\alpha)=\tau_\alpha$ for $\alpha\in
I$, then
$$\HF_{n}(\sum_{\mI}\fS)=\tau.$$

\item If for all $\alpha\in I$, $\tau_\alpha=\tau$ for some
\emph{fixed} $\tau\in \Char^n_l$, we denote $\sum_{\alpha\in\mI}
\tau_{\alpha}$ by $\tau\times\mI$.

\item If $\mI=(\{0,1\},<)$ and $\fH=\tuple{\tau_0,\tau_1}$, we
denote $\sum_{\alpha\in\mI} \tau_{\alpha}$ by $\tau_0 + \tau_1$.
\end{enumerate}
\end{definition}
The following fundamental result of Shelah can be found in
\cite{Sh75}:
\begin{thm}[Composition Theorem] \label{th:comp} Let $\vp(X_1,\dots, X_l)$ be a formula,
let $n=\qd{\vp}$ and let $\{\tau_1,\dots,\tau_m\}=\Char^n_l$.
Then,   there is a formula $\psi(Y_1,\dots , Y_m)$ such that for
every chain  $\mI=( I,<^\mI)$  and a sequence $\tuple{ \mM_\alpha
\mid \alpha\in I}$ of
 $l$-chains the following holds:
\[\sum_{\alpha\in I}\mM_\alpha\models \vp \mbox{ iff } \mI\models
\psi(Q_1, \dots Q_m), \mbox{ where }
 Q_j=\{\alpha\in I~:~ M_\alpha\models \tau_j\}\ .
\]
  Moreover, $\psi$ is computable from $\vp$.
\end{thm}

We are usually interested in cases $(2)$ and $(3)$ of the
Definition \ref{dfn:sum of n types}. The following Theorems are
important consequences of the Composition Theorem:
\begin{thm}[Addition Theorem]\label{lemma:addition of types
computable} The function which maps the  pairs of characteristic formulas  to
their sum is a recursive function. Formally, the function  $\la
n,l\in \nat. \la \tau_0, \tau_1\in \Char^n_l.\tau_0 + \tau_1$ is
recursive.
\end{thm}
\begin{thm}[Multiplication Theorem]\label{thm:typemult is recursive in MTh(I)}
Let $\mI$ be a chain. The function $\la n,l\in \nat.\la \tau\in
\Char^{n}_l.\tau \times  \mI$ is recursive in the monadic theory of
$\mI$.
\end{thm}
\subsection{The monadic theory of countable ordinals}\label{sec:ord}
B\"{u}chi (see, e.g., \cite{buchi:siefkes}) has shown that there
is a \emph{finite} amount of data concerning any countable ordinal
that determines its monadic theory.
\begin{definition}[Code of an ordinal]\label{dfn:code of an ordinal}
 Let $\alpha>0$  be a countable ordinal.
Write $\alpha=\om^\om\beta+\zeta$ where $\zeta<\om^\om$ (this can be
done in a unique way).  If $\zeta\ne 0$, write
\[\zeta = \sum_{i\le n}\om^{n-i}\cdot a_{n-i}\ , 
  \hbox{where $a_i\in \nat$ for $i\le n$ and $a_n\ne 0$}
\]
(this, too, can be done in a unique way).

 Define $\code(\alpha)$ as
$$\code(\alpha):= \left\{ \begin{array}{llr}
\tuple{0,a_n,\ldots,a_0} &   \textrm{if $\gamma=0$ }\\
\tuple{1,a_n,\ldots,a_0}     & \textrm{if $\gamma\neq 0 $ and $\zeta \neq 0$}\\

\langle 1\rangle    & \textrm{otherwise, i.e., if $\gamma\neq 0 $ and $\zeta=0$}
 \end{array} \right..$$
\end{definition}\medskip

The following is implicit in \cite{buchi:siefkes}:
\begin{thm}[Code Theorem]\label{thm:uniform decidability of MTh below om_1}
There is an algorithm that, given a sentence $\vp$ and the
\emph{code} of an ordinal  $\alpha$, determines whether
$(\alpha,<)\models \vp$.
\end{thm}

\section{Game types}\label{sec:games}
Recall that $\Char^n_2$ is the set of characteristic formulas of
the quantifier depth $n$ with free  variables among $\{X_1,X_2\}$.

For $G\subseteq \Char^n_2$ we denote by $\mG_G^\alpha$, the
{McNaughton game} $\mG_\vp^\alpha$, where $\vp$ is the disjunction
of  all formulas in $G$.

By Lemma \ref{lemma:Hintikka formulas}, for every formula
$\vp(X_1,X_2)$ of  quantifier depth $n$ there is $G \subseteq
\Char^n_2$ such that $\vp$ is  equivalent to the disjunction  of all
formulas from  $G$. Moreover, $G$ is computable from $\vp$. Hence,
in order to show that every {McNaughton game} of length $\alpha$ is
determinate, it is enough to show that for every $n$ and $G\subseteq
\Char^n_2$, the game $\mG_G^\alpha$ is determinate. Moreover, if it
is decidable who wins the games of the form $\mG_G^\alpha$,  then it
is decidable who wins $\mG_\vp^\alpha$ games.

\begin{definition}[Game Types]
Let $n\in \nat$.
\begin{enumerate}[\hbox to6 pt{\hfill}]
\item\noindent{\hskip-0 pt\bf Game type of ordinal:}\
For an ordinal $\alpha$:  $\gt_n(\alpha) $ is defined as
\[\{G\subseteq \Char^n_2\mid \mbox{ Player I wins } \mG_G^\alpha\}\ .\]

\item\noindent{\hskip-0 pt\bf Formal game-type:}\ A
formal  $n$-game-type is an element\footnote{recall that $\Ps{A}$ stands for the set of subsets of $A$}
of $\Ps{\Ps{ \Char^n_2}} $.
\end{enumerate}
\end{definition}

Let $\alpha$ be an ordinal, $ C$ be a formal $n$-game-type  and
$G\subseteq \Char^n_2$.
 We consider   the following $\alpha$-game
Game$_{\alpha}(C,G)$.
\begin{enumerate}[\hbox to6 pt{\hfill}]
\item\noindent{\hskip-0 pt\bf Game$_{\alpha}(C,G)$:}\ The game has $\alpha$ rounds and it is defined as follows:$~$

\begin{enumerate}[\hbox to6 pt{\hfill}]
\item\noindent{\hskip-0 pt\bf Round i:}\
\begin{enumerate}[$\bullet$]
\item
Player I chooses  $G_i\in C$.

\item
Player II chooses $\tau_i\in G_i$.
\end{enumerate}\medskip

\item\noindent{\hskip-0 pt\bf Winning conditions:}\
Let $\tau_i$  $(i\in \alpha)$ be the sequence that appears in the
play. I wins the play if $\Sigma_{i\in \alpha} \tau_i\in G$.
\end{enumerate}
\end{enumerate}

\noindent The following lemma is immediate:
\begin{lem} \label{lem:up-clsd} If $C_1\subseteq C_2$, $G_1\subseteq G_2$  and  I wins
Game$_{\alpha}(C_1,G_1)$, then I wins Game$_{\alpha}(C_2,G_2)$.
\end{lem}

As a consequence of Theorem \ref{th:comp} and Theorem
\ref{th:bl-game} we obtain the following lemma which will play a
prominent role in our proofs:
\begin{lem}\label{lem:comp-cond}\hfill
\begin{enumerate}[\em(1)]
\item The game Game$_{\omega}(C,G)$ is determinate.
\item It is decidable which of the players  wins Game$_{\omega}(C,G)$.
\end{enumerate}
\end{lem}
\begin{proof} We provide a reduction from Game$_{\alpha}(C,G)$
to a  McNaughton game.

Let $\{\tau_1,\dots,\tau_m\}=\Char^n_2$. For every  $G'\subseteq
\Char^n_2$: \begin{enumerate}[$\bullet$]
\item Let $J(G'):=\{j\mid \tau_j\in G'\}$.

\item Let
$\vp_{G'}(X_1,X_2)$ be $\bigvee_{\tau \in G'}\tau$ - the
disjunction of all formulas  from $G'$.
\item Let  $\psi_{G'}(Y_1,\dots ,Y_m)$ be
constructed from $\vp_{G'}$ as in the Composition Theorem (Theorem
\ref{th:comp}).
\end{enumerate}
Let $C=\{G_1,\dots, G_k\}$. Define formula $\vp_{C,G}(X_1, \dots,
X_k, Y_1, \dots, Y_m)$ as the disjunction of
 \begin{enumerate}[(1)]
\item
For all $t$ exactly one of $X_i(t)$ ($i=1,\dots, k$) holds and
$\psi_G(Y_1,\dots Y_m)$.
\item There is $t$ such that not exactly one of $Y_j(t)$ holds.
\item There is  $t$ and $i\in\{1,\dots, k\}$ such that  $X_i(t)$ and $\neg \bigvee_{j\in
J(G_i)} Y_j(t)$.
\end{enumerate}
Consider the McNaughton game $\mG_{\vp_{C,G}}^\alpha$.  The second
disjunct forces Player II at each round to assign the value  1
exactly to one of $Y_j$, and the third disjunct forces Player II to
reply to the choice of $X_i$ of Player I by choosing $Y_j$ such that
$\tau_j\in G_i$. It is clear that Player I (respectively, Player II)
has a winning strategy in Game$_{\alpha}(C,G)$ iff Player I
(respectively, Player II) has a winning strategy in
$\mG_{\vp_{C,G}}^\alpha$.

Therefore, by the B\"{u}chi-Landweber theorem, Game$_{\om}(C,G)$
is determinate, and it  is decidable whether I wins
Game$_{\omega}(C,G)$.
\end{proof}

\section{Addition and Multiplication Lemmas for Game
Types}\label{sect:comp-gt}
 This section contains the main technical
lemmas of this paper. In particular,  the role of  game types  in
Lemma \ref{lem:sum-main} and Lemma \ref{lem:det-product} is
similar to the role of the monadic types in the  Addition and
Multiplication theorems.

We say that an ordinal $\alpha$ is  \emph{determinate}
% (notations
%$Det(\alpha)$)
if for every $\MLO$ formula $\vp$, one of the
player has a winning strategy in $\mG_{\vp}^\alpha$.

In the next lemma and throughout this section we often use
``$\alpha$-game with  a winning condition $G$" for the McNaughton
game $\mG_{\vp}^\alpha$, where $\vp=\vee_{\tau\in G}\tau$;
recall that this game is also denoted by $\mG_{G}^\alpha$ (see beginning
of Section \ref{sec:games}).

\begin{lem}\label{lem:sum-main}
Let $\beta$ be an ordinal, $C_\beta=\gt_n(\beta)$ and  $G$ be a formal  $n$-game-type. For every $G_i\in C_{\beta}$, let
$K(G_i,G)=\{\tau\in \Char^n_2\mid \forall \tau'\in G_i(
\tau+\tau'\in G)\}$. Let $K(C_{\beta},G) =\cup_{G_i\in C_{\beta} }
K(G_i,G)$. If $\alpha$ and $\beta$ are determinate, then:
\begin{enumerate}[\quad\bf A:]
\item If Player I wins the  $\alpha$-game with winning condition
$K(C_{\beta},G)$, then Player I wins the $\alpha+\beta$-game with
winning condition $G$.\medskip

\item If Player I cannot win the  $\alpha$-game with winning condition
$K(C_{\beta},G)$, then Player II wins the  $\alpha+\beta$-game
with winning condition $G$.
\end{enumerate}
\end{lem}
\begin{proof}
(A) Assume that  Player I has a winning strategy $F_{\alpha} $ for
the  $\alpha$-game with winning condition $K(C_{\beta},G)$. He
plays the first $\alpha$ rounds   according to the strategy
$F_{\alpha} $. Assume that after $\alpha$ rounds the play
satisfies $\tau\in K(C_{\beta},G)$. Then there is $G_i \in
C_{\beta}$ such that $\tau \in K(G_i,G)$.  Player I plays the next
$\beta$ rounds according to a winning strategy $F_{\beta,G_i}$ for
the $\beta $ game with winning condition $G_i$. (Such a winning
strategy exists, by  definition of $C_{\beta}$.) Hence during
these $\beta$ rounds $\tau_i\in G_i$ is  realized. The type of the
play will be $\tau+\tau_i\in G$. Hence, Player I wins the
$\alpha+\beta$-game with winning condition $G$.

(B) Now assume that Player I cannot win the  $\alpha$-game with
winning condition $K(C_{\beta},G)$. Then, by determinacy of
$\alpha$ games, Player II has a winning strategy for this game.
Let him play the first $\alpha$ rounds according to this winning
strategy. Let $\tau\not\in K(C_{\beta},G)$ be the type reached
after $\alpha$ rounds.

Let $G_{\tau}=\{\tau'~:~ \tau+\tau' \in G\}$. We claim  that
Player II  wins the $\beta$-game for $G_{\tau}$. Indeed, if he
cannot win  this game, then, by determinacy of $\beta$ games,
Player I wins it, i.e., $G_{\tau}\in C_{\beta}$, and hence,
$\tau\in K(C_{\beta},G)$. Contradiction

The next $\beta$ rounds Player II  can play according to his
winning strategy for $G_{\tau}$. This will ensure that the type
$\tau'$  of this $\beta$-play is not in $G_{\tau}$. Hence, the
type of the entire play is $\tau+\tau'\not\in G$.

Hence, Player II wins the $\alpha+\beta$ game for $G$.

This completes the proof of (A) and (B).
\end{proof}
As an immediate consequence, we obtain the following Theorem:
\begin{thm} \label{lem:computabl-sum} If $\alpha$ and $\beta$ are
determinate, then \begin{enumerate}[\em(1)] \item $\alpha+\beta$ is
determinate.
\item  $\gt_n(\alpha+\beta)$ is computable from
$\gt_n(\alpha)$ and $\gt_n(\beta)$.
\end{enumerate}
\end{thm}
\proof\hfill
\begin{enumerate}[(1)]
\item is an immediate consequence of  Lemma \ref{lem:sum-main}.
\item  Assume  $C_\alpha=\gt_n(\alpha)$ and $C_\beta=\gt_n(\beta)$.
To check whether $G\in\gt_n(\alpha+\beta)$, first compute
$K(C_\beta,G)$ and then check whether $K(C_\beta,G)\in C_\alpha$.
Lemma \ref{lem:sum-main} implies that  $G\in \gt_n(\alpha+\beta)$
iff $K(C_\beta,G)\in C_\alpha$.\qed
\end{enumerate}

\noindent The  above theorem  allows us to define the addition of $n$-game
types for determinate ordinals. We will use ``+"  for the addition
of game types.

 Recall that an ordinal $\alpha$ is definable if
there is a sentence $\theta_{\alpha}$ such that for every chain
$M=\tuple{A,<}$: $M\models\theta_\alpha \mbox{ iff $M$ is
isomorphic to }\alpha.$

From the proof of Lemma \ref{lem:sum-main}  we deduce  the
following  variant of Theorem \ref{lem:computabl-sum}  for
\emph{definable strategies}:
\begin{lem} \label{lem:def-sum} Assume that $\alpha$ is definable and
in every  game of length $\alpha$ or $\beta$ one of the players has
a \emph{definable} winning strategy. Then in every $\alpha+\beta$
game one of the players has a  \emph{definable} winning strategy.
Moreover, if there are algorithms which for every $\vp$ compute a
definable winning strategy for $\mG^\alpha_\vp$ and $\mG^\beta_\vp$,
then there is an algorithm that computes a definable winning strategy
for $\mG^{\alpha+\beta}_\vp$.
\end{lem}
\begin{proof} Let $G\subseteq \Char^n_2$. We will construct a definable  winning strategy
for $\mG^{\alpha+\beta}_G$. We use here the notations of Lemma
\ref{lem:sum-main}. As shown there,  if Player I has a
winning strategy in $\mG^{\alpha+\beta}_G$, then the following
strategy $F$ is winning for Player I:
\begin{enumerate}[(1)]
\item $F$ plays the first  $\alpha$ rounds according to a winning strategy  in
$\mG^{\alpha}_{K(C_\beta,G)}$.
\item If after $\alpha$ rounds the play satisfies  $\tau\in K(G_i,G)$, then the last $\beta$ steps $F$ plays according to a
winning strategy in $\mG^{\beta}_{G_i}$.
\end{enumerate}
For $\tau\in K(C_\beta,G)$, we denote by $i_{\tau}$ the minimal
$i$ such  that $\tau\in G_i$.
 Now assume that $\psi(X_1,X_2)$ defines a winning strategy
for $\mG^{\alpha}_{K(C_\beta,G)}$ and $\chi_i(X_1,X_2)$ defines a
winning strategy for $\mG^{\beta}_{G_i}$,  and $\theta_{\alpha}$
defines the ordinal $\alpha$. Then the formula
$$\exists t \big( \theta_{\alpha}^{<t}\wedge \psi^{<t}\wedge \bigwedge_{\tau\in
K(C_\beta,G)} (\tau^{<t}\ra \chi_{i_\tau}^{\geq t})\big)$$ defines
a winning strategy described above. (Here, for a variable $t$
which does not occur free in $\vp$, we denote by $\vp^{<t}$ the
formula obtained from $\vp$ by relativizing all first order
quantifiers to $<t$, i,e. by replacing  ``$\forall u( \dots)$'' by
``$\forall u(u<t\ra \dots)$". The formula $\vp^{\geq t}$ is
defined similarly.)

The case when Player II has a winning strategy is treated
similarly.
\end{proof}

\begin{lem}\label{lem:det-product}
Let $n\in \nat$, let $\tuple{\alpha_i~:~i\in\om}$ be an
$\omega$-sequence of ordinals and let $C\subseteq {\Ps{Char^n_2}}$
be a formal $n$-game type.
 Assume that for every  $i$:
\begin{enumerate}[\em(1)]
\item For every   $G\subseteq {Char^n_2}$, the   $\alpha_i$  game for $G$  is determinate.
\item $C=\gt_n(\alpha_i)$.
\end{enumerate}
Then
\begin{enumerate}[\quad\bf A:]
\item If Player I wins
Game$_{\omega}(C,G)$, then Player I wins the  $\sum \alpha_i$-game
for $G$.\medskip

\item If Player I cannot win  Game$_{\omega}(C,G)$,  then Player II
  wins the $\sum \alpha_i$-game  for $G$.\medskip

\item For every $G\subseteq {Char^n_2}$, the  $\sum \alpha_i$ game
for  $G$ is determinate. \medskip

\item $G\in \gt_n(\sum \alpha_i)$ iff
Player I wins Game$_{\omega}(C,G)$.
\end{enumerate}
\end{lem}

\begin{proof}
(A) Let $F$ be a winning strategy for Player I in
Game$_{\omega}(C,G)$. Consider the following strategy for Player I

\begin{enumerate}[\hbox to6 pt{\hfill}]
\item\noindent{\hskip-0 pt\bf First $\alpha_0$ Rounds:}\
Let $G_0\in C$ be the first move of Player I  according to $F$.
Player I will play the first $\alpha_0$ rounds according to the
his winning strategy in the $\alpha_0$-game for $G_0$.

Let $\pi_0$ be a play according to this strategy and let
$\tau_0=type_n(\pi_0)$. Note that $\tau_0\in G_0$ and a (partial) play
$G_0\tau_0$ is consistent with $F$.\medskip

\item\noindent{\hskip-0 pt\bf Next $\alpha_{i+1}$ Rounds:}\
Let $G_{i+1}\in C$ be the  move of Player I  according to $F$
after $G_0\tau_0\dots G_i\tau_i$. Player I will play the next
 $\alpha_{i+1}$ rounds according to  his winning strategy in
 the
$\alpha_{i+1}$-game for $G_{i+1}$.

Let $\pi_{i+1}$ be a play according to this strategy during these
$\alpha_{i+1} $ rounds.

Let $\tau_{i+1}=type_n(\pi_{i+1})$. Note that $\tau_{i+1}\in
G_{i+1}$ and the  play $G_0\tau_0,\dots,$ $ G_{i+1}\tau_{i+1}$ is
consistent with $F$.
\end{enumerate}
$type_n(\pi_0\dots \pi_{i}\dots)= \Sigma \tau_i$ is in  $G$,
because the play $G_0\tau_0,\dots, G_i\tau_i,\dots $  is consistent
with the   winning  strategy $F$
 of Player I in Game$_{\omega}(C,G)$. Hence, the described
 strategy is a winning strategy for Player I in the $\sum\alpha_i $
 game for $G$.

Now let us prove (B). Assume that  Player I has no winning strategy in
Game$_{\omega}(C,G)$. By Lemma \ref{lem:comp-cond}  this game is
determinate. Hence, Player II has a winning strategy $F_2$ for
Game$_{\omega}(C,G)$.

We show that  the following strategy of Player II is winning
in the  $\sum\alpha_i$ game for $G$

\begin{enumerate}[\hbox to6 pt{\hfill}]
\item\noindent{\hskip-0 pt\bf First $\alpha_0$ Rounds:}\
For every  $D\in C$ let $\tau_D$  be the response of Player II
according to $F_2$ to the  first move $D$ of Player I in
Game$_{\omega}(C,G)$.

Let $H_0=\{\tau_D~:~D\in C\}$. We claim that Player II has a
winning strategy in  the $\alpha_0$-game for $\neg H_0$ -  the
complement of $H_0$, i.e., for ${\Char^n_2}\setminus H_0$. Indeed,
if he has no such strategy, then, by determinacy of $\alpha_0$,
Player I has a winning strategy for $\neg H_0$. Therefore, $\neg
H_0\in C$. Let $\tau=F_2(\neg H_0)$. Then $\tau\in \neg H_0$ and
$\tau\in H_0$. Contradiction.

The first $\alpha_0$ rounds Player II will play according to his
winning strategy for $\neg H_0$.

Let $\pi_0$ be a play according to this strategy and let
$\tau_0=type_n(\pi_0)$. Note that $\tau_0\in H_0$. Let $G_0$ be
such that $\tau_0=F_2(G_0)$. The (partial) play $G_0\tau_0$ is consistent
with $F_2$.\medskip

\item\noindent{\hskip-0 pt\bf Next $\alpha_{i+1}$ Rounds:}\
For every  $D\in C$ let $\tau_D$  be the response of Player II
according to $F_2$ to   $G_0\tau_0\dots G_i\tau_i D$. ($F_2$ is
defined because $G_0\tau_0\dots G_i\tau_iD$ is a play according to
$F_2$ for every $D\in C$.) Let $H_{i+1}=\{\tau_D\mid D\in C\}$. We
claim that Player II has a winning strategy in $\alpha_{i+1}$ game
for $\neg H_{i+1}$. (The arguments are the same as the arguments
which show that Player II has a winning strategy for $\neg H_0$.)

The next $\alpha_{i+1}$ rounds Player II will play according to
his winning strategy for $\neg H_{i+1}$.

Let $\pi_{i+1}$ be a play according to this strategy and let
$\tau_{i+1}=type_n(\pi_{i+1})$. Note that $\tau_{i+1}\in H_{i+1}$.
Let $G_{i+1}$ be such that $\tau_{i+1}= F_2(G_0\tau_0\dots
G_i\tau_i G_{i+1})$. The play $G_0\tau_0\dots G_{i+1}\tau_{i+1}$
is consistent with $F_2$.

\end{enumerate}
$type_n(\pi_0\dots \pi_{i}\dots)= \Sigma \tau_i$ is not in  $G$,
because the play $G_0\tau_0\dots G_i\tau_i\dots $  is consistent
with the   winning  strategy $F_2$
 of Player II in Game$_{\omega}(C,G)$. Hence, the described
 strategy is a winning strategy for Player II in $\sum\alpha_i $
 game for $G$.

 (C) and (D)  are immediate consequences of (A) and (B).
\end{proof}

As a consequence of Lemma \ref{lem:det-product} and Lemma
\ref{lem:comp-cond} we obtain the following Theorem:
\begin{thm}\label{th:mult-om} Assume  $\alpha$ is determinate.  Then
\begin{enumerate}[\em(1)]
\item $\alpha\times \om$ is determinate.
\item $\gt_n(\alpha\times\om)$ is computable from $\gt_n(\alpha)$.
\end{enumerate}
\end{thm}
From the proof of Lemma \ref{lem:det-product}, by arguments
similar to those in the proof of Lemma \ref{lem:def-sum}, we
deduce   the following  variant of Theorem \ref{th:mult-om} for
\emph{definable strategies}:
\begin{lem}\label{lem:def-prod} Assume that $\alpha$ is definable and
in every game of length $\alpha$  one of the players has a
definable winning strategy. Then in every $\alpha\times \omega $
game one of the players has a definable winning strategy.
Moreover, if there is an  algorithm which for every $\vp$ computes
a definable winning strategy for $\mG^\alpha_\vp$, then there is
an algorithm that computes a definable winning strategy for
$\mG^{\alpha\times\om}_\vp$.
\end {lem}
Note that Theorem \ref{th:mult-om} allows to define the
multiplication by $\om$ of $n$-game types for determinate
ordinals. The following Lemma allows to define $\omega$-sums of
$n$-game types of determinate ordinals. It is an analog of
Proposition \ref{prop:ordered sum preserves equivalence} for game
types.
\begin{lem}\label{lem:om-sum}
Let $n\in \nat$. Assume that \begin{enumerate}[\em(1)]
 \item $ \bar{\alpha}=\tuple{\alpha_i ~:~i\in \om}$ and $\bar{\beta}=\tuple{ \beta_i ~:~ i\in \om}$ are
$\om$-sequences of determinate ordinals  and
\item $\gt_n(\alpha_i)=\gt_n(\beta_i)$ for every $i\in \nat$.
\end{enumerate}
Then $\gt_n(\sum\alpha_i)=\gt_n(\sum \beta_i)$ and the ordinals
$\sum\alpha_i$ and  $\sum \beta_i$ are determinate.
\end{lem}
The proof of Lemma \ref{lem:om-sum} can be derived from Lemmas
\ref{lem:sum-main}, \ref{lem:det-product} and the Ramsey theorem.
Its proof is omitted, since  we won't use this lemma in the
sequel.

\section{B\"{u}chi-Landweber theorem holds for
$\alpha<\om^{\om}$}\label{sect:sho}
 In this section we provide a
simple proof of Shomrat's theorem \cite{Sho07} which extends the
{B\"{u}chi-Landweber theorem to all ordinals  $\alpha<\om^{\om}$.
\begin{thm} \label{th:shomrat} Let $\alpha<\om^{\om}$ and  $\vp(X,Y)$ be a formula. Then:
\begin{enumerate}[\hbox to6 pt{\hfill}]
\item\noindent{\hskip-0 pt\bf Determinacy:}\ One of the players has a
  winning strategy in the game $\mG_\vp^\alpha$.\medskip

\item\noindent{\hskip-0 pt\bf Decidability:}\ It is decidable
  \emph{which} of the players has a winning strategy.\medskip

\item\noindent{\hskip-0 pt\bf Definable strategy:}\ The player who has
  a winning strategy also has a \emph{definable} winning
  strategy.\medskip

\item\noindent{\hskip-0 pt\bf Synthesis algorithm:}\ We can compute a formula $\psi(X,Y)$ that defines (in $(\alpha,<)$)
a winning strategy for the winning player in $\mG_\vp^\alpha$.
\end{enumerate}
\end{thm}
\begin{proof}
Note that every ordinal $\alpha<\omega^{\omega}$ is definable.
Games of length one have  definable winning strategies.

 First,
prove by  the induction on $n\in\nat$ that the theorem holds for
$\alpha=\om^n$.
 The base $\alpha=1$ is trivial. For
the inductive step use Lemma \ref{lem:def-prod}.

Next, by Lemma \ref{lem:def-sum}, we obtain that the theorem is
true for every $\alpha$ of the form $\om^{m}n_m+
\om^{{m-1}}n_{m-1} +\nek+ +\om^0 n_0$, where $m,n_i\in \nat$.

Finally, note that every $\alpha<\om^\om$ is equal to $\om^{m}n_m+
\om^{{m-1}}n_{m-1} +\nek+ +\om^0n_0$, for some  $k,n_i\in \nat$.
\end{proof}
In \cite{RS06}  we proved that for every  $\alpha\geq \om^\om$ there
is a formula $\psi_\alpha$ such that Player I wins
$\mG_{\psi_{\alpha}}^\alpha$; however, he has no definable winning
strategy in this game, i.e., the definability part of the
{B\"{u}chi-Landweber theorem fails for every $\alpha\geq\om^\om$.
Therefore, the {B\"{u}chi-Landweber theorem holds for $\alpha$ iff
$\alpha<\om^\om$.

\section{Determinacy}\label{sect:deter}
Theorems   \ref{lem:computabl-sum} and \ref{th:mult-om}
  imply that the set of determinate
ordinals is closed under addition  and multiplication by $\om$. In
this section, we prove that every countable ordinal is
determinate. First, let us show the following Lemma:
\begin{lem}\label{lem:det-om-power} $\omega^{\alpha}$ is determinate for every countable
$\alpha$.
\end{lem}
\begin{proof}
By Induction on $\alpha$.

The basis: $\alpha=0$ is immediate.

The case of successor  follows from Lemma \ref{th:mult-om}(1).

Assume that $\alpha$ is a countable  limit ordinal. In this case
$\alpha=\lim_{i\in \nat}\beta_i$, where $\beta_i<\alpha$ is an
increasing $\omega$-sequence.

We  are going to show that for every $n\in \nat$ and every
$G\subseteq {Char^n_2}$, the $\omega^{\alpha}$ game for $G$ is
determinate.

For every $n$,  the set $\Ps{\Char^n_2}$ is finite. Therefore,
there is $C$ and  an increasing $\omega$-subsequence $\gamma_i$ of
$\beta_i$ such that \begin{center}
 $C=\gt_n(\om^{\gamma_i})$ for every $i$.
\end{center}
Let $\alpha_i=\omega^{\gamma_i}$.  By the inductive assumption
$\alpha_i$ are determinate. Therefore, by Lemma
\ref{lem:det-product}(C) and the equation above, for every
$G\subseteq {Char^n_2}$, the $\sum \alpha_i$ game for $G$ is
determinate.

Note that $\sum\alpha_i= \sum
\omega^{\gamma_i}=\omega^{\lim\gamma_i}=\omega^{\lim\beta_i}=\omega^{\alpha}$.
Hence, for every $n\in \nat$ and every $G\subseteq {Char^n_2}$,
the $\omega^{\alpha}$ game for $G$ is determinate. Therefore,
$\omega^{\alpha}$ is determinate.
\end{proof}
\begin{thm} \label{th:determ} Every countable   ordinal is determinate.
\end{thm}
\begin{proof} Every  ordinal has a Cantor Normal Form  representation
$\omega^{\alpha_1}n_1+ \omega^{\alpha_2}n_2 + \cdots
+\omega^{\alpha_k}n_k$, where  $k,n_i\in \nat$. Hence, by Lemmas
\ref{lem:det-om-power} and \ref{lem:computabl-sum}, every
countable ordinal  is determinate.  \end{proof}

\section{Decidability} \label{sect:dec}
\begin{lem}\label{th:dec-mult} For every $n\in\nat$ there is
$m\in\nat$ such that $\gt_n(\om^m)=\gt_n(\om^m\times \alpha)$ for
every countable ordinal $\alpha>0$. Moreover, $m$ is computable
from $n$.
\end{lem}
\begin{proof}  The proof is similar to the proof of Theorem 3.5(B) in \cite{Sh75}.
For every ordinal $\alpha$, let us write $t(\alpha)$ for
$\gt_n(\alpha)$. By Lemma \ref{th:mult-om} there is a
multiplication by $\om$ operation  on the game-types of
determinate ordinals. We denote it by $\times \om$; for every
determinate ordinal $\alpha$ if  $C=t(\alpha)$, then $C\times \om$
is $t(\alpha\om))$. The multiplication by $\om$ operation is even
computable by Lemma \ref{th:mult-om}.
 Moreover, If all $\alpha_i $
($i\in \om$) are determinate and $C=t(\alpha_i)$ for all $i\in
\om$ and  \emph{fixed} $C$, then $t(\sum \alpha_i)=C\times \om$.

Let $m:=|\Ps{\Ps{\Char^n_2}}|$, i.e., the number of possible
formal $n$-game types. We are going to show that this $m$ satisfies the Lemma.

 We first show that there is a $p<m$
such that $t(\om^p)=t(\om^{p+1})$.

By the pigeon-hole principle, there are $q<r\le m$ such that
$t(\om^q)=t(\om^r)$. Moreover, by Theorem \ref{th:mult-om} we may
compute such $q$ and $r$. If $r=q+1$, our claim is proved for
$p=q$, so assume $r\ge q+2$. We will show that $p=q+1$ works. By
Theorem \ref{th:mult-om}(2),
$$t(\om^{q+2}) = t(\sum_{i\in \nat}(\om^{q+1} + \om^q))=
\big(t(\om^{q+1}+\om^q)\big)\times \om =$$
$$ \big(t(\om^{q+1})+t(\om^q)\big)\times \om =\big(t(\om^{q+1})+t(\om^r)\big)\times
\om =$$$$\big(t(\om^{q+1}+\om^r)\big)\times \om.$$
 But,
$q+1<r$, so $\om^{q+1} + \om^r = \om^r$. Thus, indeed,
$$t(\om^{q+2}) = t(\om^r)\times \om=t(\om^q)\times \om =
t(\om^{q+1}).$$

Next, we show that any $p$, as above, will satisfy the claim of
the lemma, i.e., $t(\om^p)=t(\om^p\alpha)$ for every countable
$\alpha>0$. First, note that
$$
\begin{array}{lr}
t(\om^p)+t(\om^p)=t(\om^p) + t(\om^{p+1}) = &\\
&\quad (\star\star)\\
 t(\om^p + \om^{p+1})=t(\om^{p+1})=t(\om^p).

\end{array} $$ Now, by induction on countable
$\alpha>1$, we prove $t(\om^p\cdot \alpha)=t(\om^p)$. Assume
$\alpha$ is a successor, say, $\alpha = \beta+1$. If $\beta=0$,
there is nothing to prove. Assume $\beta>0$.
$$t(\om^{p}\alpha) = t(\om^{p}\beta + \om^{p}) = t(\om^{p}\beta) + t(\om^{p})
=t(\om^{p}) + t(\om^{p}) \stackrel{(\star\star)}{=} t(\om^p).$$

Finally, assume that $\alpha$ is a limit ordinal. Then there is an
increasing  $\omega$-sequence $\tuple{\alpha_i~:~i\in \om}$ such
that $\alpha =\lim \alpha_i$. Hence,
\begin{equation}\label{eq:7}
\omega^p\times \alpha=\lim (\omega^p\times \alpha_i)=\sum
\omega^p(\alpha_{i+1}-\alpha_i).
\end{equation}

  By
the inductive assumption $t(\om^p)=t(\omega^p\times
(\alpha_{i+1}-\alpha_i))$. Hence, $t(\omega^{p+1})=t(\om^p)\times
\om =t(\sum_{i\in\om} \om^p\times
(\alpha_{i+1}-\alpha_i))=t(\om^p\alpha)$, by Lemma
\ref{lem:det-product}. Since $t(\om^p)=t(\om^{p+1}$), we derive
that $t(\om^p)=t(\om^p\times \alpha)$.

Recall that $m:=|\Ps{\Ps{\Char^n_2}}|\geq q+1=p$. Therefore,
$t(\om^m)=t(\om^{p}\om^{m-p})$ $=t(\om^{p}\om^{m-p+1})=t(\om^{m+1})$.
 Hence, $m$ satisfies the lemma.
\end{proof}
Let $n\in \nat$ and let $m=m(n)\in\nat$ be computable from $n$ as
in Lemma \ref{th:dec-mult}.
 Every ordinal $\alpha>0$ has a unique representation of the
form $\alpha=\om^m\gamma+ \omega^{m-1}n_1+ \omega^{m-2}n_2 +
\cdots +\omega^{0}n_m $, where $n_i \in \nat$.
 Define $\code_n(\alpha)$ as
$$\code_n(\alpha):= \left\{ \begin{array}{lll}
\tuple{0,n_1,n_2,\dots,n_m} & \textrm{if $\gamma=0$ }\\
\tuple{1,n_1,n_2,\dots,n_m}     & \textrm{otherwise}
 \end{array} \right..$$

 \begin{thm}[Decidability]\label{thm:dec-ch-prob}
There is an algorithm that given a formula $\vp(X_1,X_2)$ and the
$\code_{\qd{\vp}}(\alpha)$  of a countable ordinal $\alpha>0$,
determines which of the players has a winning strategy in
$\mG_\vp^\alpha$.
\end{thm}
\begin{proof}
 By Lemma \ref{th:dec-mult}, we have that $\gt_n(\om^m)=\gt_n(\om^m\times \gamma)$ for every $\gamma>0$.
 Therefore, by Lemma \ref{lem:computabl-sum}, if $\code_n(\alpha)=\code_n(\beta)$, then
 $\gt_n(\alpha)=\gt_n(\beta)$.

 Hence, for an input  $\vp$ and $\code_{\qd{\vp}}(\alpha)=\tuple{n_0,n_1,n_2,\dots,n_m}
 $ it is enough to decide
    which of the players wins
 the $\mG_\vp^{\om^mn_0+\om^{m-1}n_1+\dots +\om^0n_m}$. This is decidable  by Theorem
 \ref{th:shomrat}.
 \end{proof}
\section{MLO Characterisation of the Winner} \label{sect:mlo-char}
In this section we show that for every
formula $\varphi(X,Y)$  there exists
an  $\MLO$ sentence $\psi$  such that  for every   countable
ordinal  $\alpha$:  Player I wins $\mG_\vp^\alpha$ iff $
\alpha\models \psi$.
%
%that holds exactly for the  countable
%ordinals such that Player I wins $\alpha$-game with the winning condition $\varphi$.
%
%Section \ref{sect:mlo-char} provides an  MLO-characterisation of the winner
%and  shows that for every
%formula $\varphi(X,Y)$  there exists
%an  $\MLO$ sentence $\psi$   such that  for every   countable
%ordinal  $\alpha$:  Player I wins $\mG_\vp^\alpha$ iff $
%\alpha\models \psi$.
%

Our proof refines the proof of Theorem \ref{thm:dec-ch-prob}.
For every  $n\in \nat $ we will define a mapping $\gcode_n$ which assigns to every countable ordinal a code (game code).
Then, we show that (1) the  $\gcode_n(\alpha)$ determines the $\gt_n(\alpha)$ (2) for every code $c$ in the range of  $gcode_n$,
 the set $ A_c:=\{\alpha \mid \gcode_n(\alpha)=c\}$ is
definable by an $\MLO$ sentence and (3) the range of $\gcode_n$ is a finite set.
From   (1)-(3) it will be easy to show that
the set of countable ordinals for which Player I  has a winning strategy in $\mG_\vp^\alpha$ is $\MLO$ definable.

Let $\alpha=\om^mn_0+ \omega^{m-1}n_1+ \omega^{m-2}n_2 +
\cdots +\omega^{0}n_m $. Theorem \ref{thm:dec-ch-prob} shows that in order  to determine  $\gt_n(\alpha)$ we do not have to know the coefficients of
the $\om^i$ for every   sufficiently large $i$.
The next Lemma will imply that in order to determine  $\gt_n(\alpha)$  we do not need to know the precise value of the coefficients $n_i$,  even  for small  values of
$i$. It is sufficient to know these coefficients modulo a  number which depends on $n$.

Recall that $\alpha n$ denotes the multiplication of the ordinal $\alpha$ by $n$.

\begin{lem}\label{lem:fin-mult} For every $n\in\nat$ there is
$p\in\nat$ such that if $n_1>n_2\geq p$ and $n_1=n_2~mod ~p$, then $\gt_n(\alpha n_1)=\gt_n(\alpha n_2)$ for
every  ordinal $\alpha>0$. Moreover, $p$ is computable
from $n$.
\end{lem}
\begin{proof}
For every ordinal $\alpha$, let us write $t(\alpha)$ for
$\gt_n(\alpha)$.

Let
$m:=|\Ps{\Ps{\Char^n_2}}|$, i.e., the number of possible
formal $n$-game types.

By the pigeon-hole principle, there are $q<r\le m+1$ such that
$t(\alpha q)=t(\alpha r)$.
By Theorem \ref{lem:computabl-sum}, for every $n$:
$$
t(\alpha(n+q))= t(\alpha n) + t(\alpha q) = t(\alpha n) + t(\alpha r)= t(\alpha(n+r))$$
Therefore, if $n_1>n_2\geq q$ and $n_1$ is equal to $n_2$ modulo $r-q$, then $t(\alpha n_1)=t(\alpha n_2)$.
 Note that $q$ and $r$ depend on $\alpha$. However, if we define $p:=(m+1)!$, then this $p$  satisfies the lemma.
\end{proof}

Now we are ready to define $\gcode_n$.
Let $n\in \nat$ and let  $p=p(n)\in \nat $ be computable from $n$ as
in Lemma \ref{lem:fin-mult}.
Let $ \trun_n:\nat\ra \nat$ be defined as
$$trun_n(k):=
\left\{ \begin{array}{lll}
k & \textrm{if $k<p$ }\\
p+ (k ~mod~ p)      & \textrm{otherwise}
 \end{array} \right..$$

Let $m=m(n)\in\nat$ be computable from $n$ as
in Lemma \ref{th:dec-mult}.
 Every ordinal $\alpha>0$ has a unique representation of the
form $\alpha=\om^m\gamma+ \omega^{m-1}n_1+ \omega^{m-2}n_2 +
\cdots +\omega^{0}n_m $, where $n_i \in \nat$.

 Define $\gcode_n(\alpha)$ as
$$\gcode_n(\alpha):= \left\{ \begin{array}{lll}
\tuple{0, \trun_n(n_1), \trun_n(n_2),\dots,\trun_n(n_m)} & \textrm{if $\gamma=0$ }\\
\tuple{1, \trun_n(n_1), \trun_n(n_2),\dots,\trun_n(n_m)}   & \textrm{otherwise}
 \end{array} \right..$$
(Hence, the tuple $\gcode_n(\alpha) $ is obtained from $\code_n(\alpha)$ by applying $trunc_n$ pointwise.)
 By Theorem \ref{lem:computabl-sum}, Lemma \ref{th:dec-mult} and Lemma \ref{lem:fin-mult} it follows that $\gt_n(\alpha)$ is determinate by
 $\gcode_n(\alpha)$. Moreover, by  a proof similar to the proof of Theorem \ref{thm:dec-ch-prob}, we obtain the following result:
\begin{lem}[Game Code Lemma]\label{lem:gcode}
There is an algorithm that given a formula $\vp(X_1,X_2)$ and the
$\gcode_{\qd{\vp}}(\alpha)$  of a countable ordinal $\alpha>0$,
determines which of the players has a winning strategy in
$\mG_\vp^\alpha$.
\end{lem}
\begin{lem}\label{lem:gcode-def}
For every $c$ in the range of $\gcode_n$ there is a sentence $\psi_c$ such that
$\alpha \models \psi_c$ iff  $\gcode_n(\alpha)=c$.
\end{lem}
\begin{proof}(Sketch) For every $i\in\nat$ there is an $\MLO$ formula $Mult_{\om^i}(X)$  which says that $X$ is a subset of the order
type $\om^i\gamma$ for an ordinal  $\gamma>0$.
For every $k<p\in \nat$ there is an $\MLO$  formula $Mod_{p,k}(X,Y)$ which says that $X$ contains $k ~mod ~p$ occurrences of $Y$.
Using these formulas it is not difficult to formalize the definition of $\gcode_n$  and to write a desirable formula $\psi_c$.
\end{proof}
Now we are ready to provide an $\MLO$ characterization of  the winner.
\begin{thm}
There is an algorithm that given a formula $\vp(X_1,X_2)$  computes a sentence $Win_\varphi$ such that for every countable ordinal $\alpha$:
Player I wins  $\mG_\vp^\alpha$ if and only if $\alpha \models Win_\vp$.

\end{thm}
\begin{proof}
Let $n$ be the quantifier depth of $\vp$.
Let $C_\vp$ be defined as
$$C_\vp:=\{c \mbox{ is a }\gcode_n \mid \mbox{ Player I wins } \mG_\vp^\alpha \mbox{~if }\gcode_n(\alpha)=c  \}$$
Since the range of $\gcode_n$ is finite, the set $C_\vp$ is finite and computable by Lemma \ref{lem:gcode}.
For every $c\in C_\vp$   compute $\psi_c$ as in Lemma \ref{lem:gcode-def}.
Hence,   $Win_\vp$  defined as
$$Win_\vp:=\bigvee_{c\in C_\vp} \psi_c$$
satisfies the lemma.
\end{proof}
\section{Synthesis Problem}\label{sect:synth}
  Let $\alpha$ be an ordinal.
  In this section we address the following  synthesis problem:

\begin{enumerate}[\hbox to6 pt{\hfill}]
\item\noindent{\hskip-0 pt\bf Problem Synth($\alpha$):}\
\begin{enumerate}[\hbox to6 pt{\hfill}]
\item\noindent{\hskip-0 pt\bf Input:}\
   a formula $\vp(X_1,X_2)$,
\item\noindent{\hskip-0 pt\bf Task:}\ decide whether one of the players  has a definable winning strategy in
$\mG_\vp^\alpha$, and if so, construct $\psi$ which defines his
winning strategy. \end{enumerate}\end{enumerate} The decidability version of
the synthesis problem for $\alpha$ requires  only to decide
whether one of the players has a definable winning strategy in
$\mG_\vp^\alpha$ (but does not output it). We will  denote  this version by
Dsynth($\alpha)$.

By Theorem \ref{th:shomrat}, these problems are computable  for
$\alpha<\om^\om$. Note Dsynth($\alpha)$ is trivial  for
$\alpha<\om^\om$, because for these ordinals the winning player
has a definable winning strategy.

As mentioned in the introduction, for every  $\alpha\geq \om^\om$
there is a formula $\psi_\alpha$ such that Player I wins
$\mG_{\psi_{\alpha}}^\alpha$; however, he has no definable winning
strategy. Therefore, Dsynth($\alpha)$ is non-trivial for
$\alpha\geq\om^\om$.

Unfortunately, we were unable to show that the synthesis problems
are computable for every countable ordinal. However, we   show here
that a crucial ordinal is $\om^\om$.

First, note
\begin{lem}\label{lem:win}  There is an algorithm which for formulas
$\psi(X_1,X_2) $ and $\vp(X_1,X_2)$ constructs sentences
$\win_I^{\vp,\psi}$ and $\win^{\vp,\psi}_{II}$ such that for every
ordinal $\alpha$:
\begin{enumerate}[\em(1)]
\item $\alpha \models \win_I^{\vp,\psi}$ iff $\psi$ defines a winning strategy for Player I in
$\mG^\alpha_\vp$.
\item $\alpha\models \win_{II}^{\vp,\psi}$ iff $\psi$ defines a winning strategy for Player II in
$\mG^\alpha_\vp$.
\end{enumerate}
\end{lem}
\begin{proof}
$\win_I^{\vp,\psi}$ is the conjunction of the sentence that says
that $\psi$ defines a strongly causal operator and the sentence
$\forall X_2X_1\psi(X_1,X_2)\ra \vp(X_1,X_2)$.

$\win_{II}^{\vp,\psi}$ is defined similarly.
\end{proof}

Recall  (see Section \ref{sec:ord}) that  the monadic theory of a
countable ordinal $\alpha$ is definable from Code$(\alpha)$. From
Lemma \ref{lem:win} we deduce:
\begin{lem}\label{lem:code-red-synth}\hfill
\begin{enumerate}[\em(1)]
\item
The problems Synth($\alpha)$ and Dsynth($\alpha)$ are recursive in
each other. \item If Code$(\alpha_1)$=Code$(\alpha_2)$, then the
problems Synth($\alpha_1)$ and Synth($\alpha_2)$ are equivalent,
and the problems Dsynth($\alpha_1)$ and Dsynth($\alpha_2)$ are
equivalent.
\end{enumerate}
\end{lem}
\proof\hfill
\begin{enumerate}[(1)]
\item
It is clear that Dsynth($\alpha)$ is recursive in Synth($\alpha$).
We will show that Synth($\alpha)$ is recursive in
Dsynth($\alpha$). Let $\psi_1,\dots \psi_i\dots$ be a recursive
enumeration of all formulas. Let $\vp$ be an input for
 Synth($\alpha$). First, check   Dsynth$(\alpha)$  on  input $\vp$.
 If the answer is ``No", then neither of the players has a definable
 winning strategy. If the answer is ``Yes", then for $i=1,\dots$ use Lemma
 \ref{lem:win} and Theorem \ref{thm:uniform decidability of MTh below om_1}
 to verify whether
  $\psi_i$ is a definable winning strategy for one of the players. When  such $\psi_i$ is found,
  output it and terminate.
The correctness of the reduction is immediate.

 \item If  Code$(\alpha_1)$=Code$(\alpha_2)$, then by Theorem \ref{thm:uniform decidability of MTh below
 om_1}, $\alpha_1$ and $\alpha_2$ satisfy the same monadic
 sentences. Hence, by Lemma \ref{lem:win}, $\psi$ is a definable winning strategy
 in $\mG^{\alpha_1}_\vp$ iff $\psi$ is a definable winning strategy
 in $\mG^{\alpha_2}_\vp$.\qed
\end{enumerate}

The following lemma provides a reduction from Dsynth$(\alpha)$ to
Dsynth$(\om^\om)$}. Hence,  Synth$(\alpha)$ is reducible to
Synth$(\om^\om)$.

\begin{lem}
 \label{lem:red=todecid} There is an algorithm that given the $Code(\alpha)$ of a  countable  $\alpha>\om^\om$
 and $\vp_1(X_1,X_2)$ constructs $\vp_2(X_1,X_2)$ such that
a player  has a definable winning strategy in
$\mG_{\vp_2}^{\om^\om}$ iff  he  has a definable winning strategy
in $\mG_{\vp_1}^\alpha$.
\end{lem}

\begin{proof}
Let $\alpha=\om^\om\alpha' + \beta$ with $\beta<\om^\om$ (this can
be done in a unique way), and let  $\alpha_1= \om^\om + \beta$.

Code($\alpha$)=Code($\alpha_1)$, therefore, by lemma
\ref{lem:code-red-synth},  $\psi$ defines a winning strategy in
$\mG_{\vp_1}^\alpha$ iff it defines a winning strategy in
$\mG_{\vp_1}^{\alpha_1}$.

 Let $n$ be the  quantifier depth of $\vp_1$.  Compute $G\subseteq \Char^n_2$
such that $\vp_1$ is equivalent to the disjunction of formulas
from $G$. Let $C_\beta$ and $K(C_\beta,G)$ be defined as in Lemma
\ref{lem:sum-main}. Let $\vp_2$ be the disjunction of formulas
from $K(C_\beta,G)$. Note that $\vp_2$ is computable from
code($\alpha)$ and $\vp_1$.

 We claim that $\vp_2$ satisfies the conclusion of the Lemma.

 Indeed, from the proof of Lemma \ref{lem:sum-main}, it follows
 that  if $F$ is a winning strategy in $\mG_{\vp_1}^{\om^\om+\beta}$,
 then its first  $\om^\om$ rounds is a winning strategy
 $\mG_{K(C_\beta,G)}^{\om^\om}=\mG_{\vp_2}^{\om^\om}$.

Assume that  $\psi(X_1,X_2)$ defines $F$ in $\om^\om+\beta$. We
are going to construct $\psi_2$ which defines  in $\om^\om$ the
strategy $F_1$ which plays $F$  first $\om^\om$ rounds.

Let $\Delta =\{\tuple{\tau,\tau'}\in\Char^n_2\times
\Char^n_2~:~\tau+\tau'\ra \psi\}$.  Define  $\Delta_1$ as $\Delta_1:=\{\tau~:~
\exists\tau' ~:~\tuple{\tau,\tau'}\in \Delta \mbox{ and }
\beta\models \exists X_1 X_2 \tau'\}$. Note that $\Delta$ and
$\Delta_1$ are computable.

From Theorem \ref{lemma:addition of types computable}, it follows
 that the disjunction of formulas from $\Delta_1$ defines $F_1$.

For the other direction, assume that there is a definable winning
strategy  $F_1$ for  Player I in   $\mG_{K(C_\beta,G)}^{\om^\om}$.
Note that for every $\tau\in {K(C_\beta,G)}$ there is a winning
strategy for Player I  in the $\beta$ game for
$H_{\tau}=\{\tau'~:~\tau+\tau'\in G\}$. Since $\beta<\om^\om$
Player I has a definable winning strategy $F_\tau$ in the $\beta$
game for $H_\tau$. The definable winning strategy for
$\mG_{G}^{\om^\om+\beta}$ is constructed from $F_1$ and the
strategies $F_\tau$ as in Lemma \ref{lem:def-sum}. The only subtle
point is that the ordinal $\om^\om$ is not definable. However,
there is a formula $\theta(t)$ such that
$\alpha_1\models\theta(\gamma)$ iff $\gamma=\om^\om$ (this formula
expresses  that $\gamma$ is the minimal ordinal such that the
interval $[\gamma,\alpha_1)$ is isomorphic to (a definable
ordinal) $\beta$.

  The case when
there is a definable winning strategy   for  Player II in
$\mG_{K(C_\beta,G)}^{\om^\om}$ is similar.
\end{proof}

\section{Conclusion and Open Problems}\label{sect:concl}
We considered a natural extension of the Church Problem to
countable ordinals. We proved that the  B\"{u}chi-Landweber
theorem extends fully to all ordinals $<\om^\om$ and that its
determinacy and decidability  parts extend to all countable
ordinals. We reduced the synthesis problem for countable  ordinals
$>\om^\om$ to the synthesis problem for $\om^\om$. However, the
decidability of the synthesis problems for $\om^\om$ remains open.

In preliminary version of this paper we asked  whether the first
uncountable ordinal $\om_1$ is determinate.
For uncountable ordinals the situation changes radically.
Let $\vp_{\mathit{spl}}(X,Y)$ say:
``$X$ is stationary, $Y\s X$ and both $Y$ and $X\sm Y$ are stationary''
(recall  that  $S\s \om_1$ is called \emph{stationary} iff for every closed unbounded  $C\s \om_1$, $S\cap C\neq \emptyset$).
P. B. Larson and  S. Shelah pointed to us
that it follows immediately from \cite{larson+shelah} that each of the following statements is consistent
with ZFC:
\begin{enumerate}[(1)]
\item None of the players has a winning strategy in $\mG_{\vp_{\mathit{spl}}}^{\om_1}$.
\item Mrs.\ $Y$ has a winning strategy in $\mG_{\vp_{\mathit{spl}}}^{\om_1}$.
\item Mr.\ $X$ has a winning strategy in $\mG_{\vp_{\mathit{spl}}}^{\om_1}$.
\end{enumerate}
In other words, ZFC can hardly tell us anything concerning this game.
On the other hand, S.\ Shelah \cite{She07} tells  us he believes it should be possible to prove:
\begin{conj}\label{conj:consistent that all om_1 games are determined}
It is consistent with ZFC that $\mG_\vp^{\om_1}$ is determined for \emph{every} formula $\vp$.
\end{conj}
%
%Shelah \cite{She07}
%proved
%\begin{theorem} The determinacy of  $\om_1$ games is independent
%from ZFC.
%\end{theorem}
Let us discuss a question of uniform definability of the winning strategy.
Recall that for every $\vp$ and $\alpha< \om^\om$ one of the players has a finite memory winning strategy in  $\mG_\vp^\alpha$.
It is natural to ask the following {uniform definability} question:
given  $\vp(X,Y)$  as above, is it possible to provide
$\psi$  such that, for each ordinal below $\omega^\omega$  such that
Player I wins  $\mG_\vp^\alpha$,  the formula $\psi$  defines his winning strategy in $\mG_\vp^\alpha$?

The negative answer to the uniform definability question even for one-player games   follows from
our results in \cite{RS07}.
Consider  the formula $\vp_{\om\ub}(Y)$  expressing that ``$Y$ is an unbounded $\om$-sequence.''
It is clear that the moves of the Player I are unimportant in the  games with the  winning condition $\vp_{\om\ub}(Y)$, and that Player II wins this game over every countable limit ordinal.
However, it was shown in \cite{RS07} that Player II has no definable winning strategy which uniformly works for all limit ordinals $<\om^\om$.

The negative answer to   the {uniform definability} question leads to the following algorithmic problem:
\begin{enumerate}[\hbox to6 pt{\hfill}]
\item\noindent{\hskip-0 pt\bf Problem:}\ (uniform definability of
  winning strategy) Given $\vp(X,Y)$. Decide whether there is $\psi$
  which defines a winning strategy for Player I for each ordinal below
  $\omega^\omega$ such that Player I wins the $\mG_\vp^\alpha$.
\end{enumerate}
The decidability of {uniform definability problem} is an open question.

Next, we describe a uniformization problem, sometimes called the
Rabin uniformization problem.

Let $\vp(\bar{X},\bar{Y})$, $\psi(\bar{X},\bar{Y})$ be formulas
and $\mM$  be a structure. We say that $\psi$ \emph{uniformizes}
(or, is a \emph{uniformizer} for) $\psi$ \emph{in} $\mM$ iff:
\begin{enumerate}[(1)]
\item $\mM \models \forall \bar{X}\exists^{\le 1}\bar{Y}\psi(\bar{X},\bar{Y})$,
\item $\mM \models \forall\bar{X}\forall \bar{Y}(\psi(\bar{X},\bar{Y})\rar \vp(\bar{X},\bar{Y}))$, and
\item $\mM \models \forall\bar{X}\big(\exists \bar{Y}\vp(\bar{X},\bar{Y}) \rar \exists \bar{Y}\psi(\bar{X},\bar{Y})\big)$.
\end{enumerate}
$\mM$ has the \emph{uniformization property} iff every formula
$\vp$ has a uniformizer $\psi$ in $\mM$.

 In
\cite{lifsches+shelah}, Lifsches and Shelah show that an ordinal
$\alpha$ has the uniformization property iff $\alpha<\om^\om$.

Uniformization, too, naturally leads to a decision problem:

\noindent $\mbox{}$\hspace{0.1cm} \framebox {
\begin{minipage}{0.95\hsize}
\noindent \hspace{2.0cm}{\em Uniformization  Problem} for $\alpha$:\\
\noindent{\em Input:} an  $MLO$ formula $\vp(X,Y)$.\\
\noindent{\em Task:} determine whether $\vp$ has a uniformizer in
$\alpha$, and if so, construct  it.

\end{minipage}
}\\

 Note the
similarities and dissimilarities between the Church synthesis
problem (see Def.  \ref{dfn:Church synthesis problem}) and the
uniformization problem. In uniformization, we are also given a
formula $\vp(X,Y)$ and to every $P$ we try and ``respond'' with a
$Q$, such that $\vp(P,Q)$ holds. Only we do not restrict ourselves
to causal responses. On the other hand, we do restrict ourselves
to \emph{definable} (in $(\alpha,<)$) responses. In the Church
problem, we do \emph{not} require that the strategy (=causal
operator) is definable.

While we are not yet able to decide uniformization in
$(\om^\om,<)$, we presented in \cite{RS07}  a restricted version
of this problem, and proved  that this version is  decidable for
every countable ordinal.

 Our initial motivation to study games of length $>\om$ was  a
 hope to reduce $\omega$-games with complex winning conditions \cite{CDT02,Se06} to
 longer games with simple winning conditions. We plan to pursue this direction further.

\section*{Acknowledgments}
I am very grateful to  Amit Shomrat for his insightful comments.

\end{document}